\documentclass[twoside,twocolumn,english]{article}

\usepackage[T1]{fontenc}
\usepackage[latin9]{inputenc}
\usepackage[letterpaper]{geometry}
\usepackage{amsthm}
\usepackage{amsmath}
\usepackage{graphicx}
\usepackage{amssymb}
\usepackage{epstopdf}
\makeatletter

\theoremstyle{plain}

  \theoremstyle{definition}
  
  \theoremstyle{plain}
  
 \theoremstyle{definition}
 \newtheorem*{defn*}{Definition}

\makeatother

\usepackage{babel}

\begin{document}

\title{{\Large \bf Incentive Compatible Influence Maximization in Social
Networks and Application to Viral Marketing}}

\author{Mayur Mohite%
\thanks{Department of Computer Science and Automation, Indian Institute
of Science, mail - mayur@csa.iisc.ernet.in%
} ~and~Y. Narahari%
\thanks{Professor, Department of Computer Science and Automation, Indian Institute of
Science, mail - hari@csa.iisc.ernet.in%
}~~}
\maketitle
\begin{abstract}
Information diffusion and influence maximization are important
and extensively studied problems in social networks. 
Various models and algorithms have been proposed in the literature
in the context of the influence maximization problem. 
A crucial assumption in all these studies is 
that the influence probabilities are known to the social planner. 
This assumption is unrealistic since the influence probabilities
are usually private information of the individual agents and
strategic agents may not reveal them truthfully.
Moreover, the influence probabilities could vary significantly with the type 
of the information flowing in the network and the time at which the information 
is propagating in the network. 
In this paper, we use a mechanism
design approach  to elicit influence probabilities
truthfully from the agents. 
We first work with a simple model, the {\em influencer model\/}, where we
assume that each user knows the level of influence she has on her
neighbors but this is private information. In the second model,
the {\em influencer-influencee model\/}, which is more realistic, we
determine influence probabilities by combining the probability values
reported by the influencers and influencees.
In the context of the first model, we present how  
VCG  (Vickrey-Clarke-Groves) mechanisms could be used for truthfully
eliciting the influence probabilities.
Our main contribution is to design 
a scoring rule based mechanism in the context of the
influencer-influencee model. In particular, we show the incentive
compatibility of the mechanisms when the scoring rules are proper
and propose a reverse weighted scoring rule based mechanism as
an appropriate mechanism to use. We also discuss briefly the
implementation of such a mechanism in viral marketing applications.
\end{abstract}

\section{Introduction}
Social networks are widespread and indeed provide an
effective medium to propagate information and to market
and advertise products.
Examples of online social networks include 
facebook, twitter, linked in, orkut, etc. 
A social network, in a natural way, could be represented in the form of  a graph in 
which an individual is represented as a
node and there is an edge between two individuals if they are associated
with each other. 

Consider a situation in which a company has designed a new gaming
console which it wants to market on a social network. As a marketing
strategy, the company can select a small set of influential
users to whom it provides the product for free. If these users like the product, 
then they will recommend it to their friends. These friends could 
get influenced by the users and
will perhaps buy the product. Some of  these friends
will in turn recommend the product to their
friends, leading to a marketing cascade. 
An online social network is an effective medium
for launching such a marketing campaign because it has much information
about the users as well as the relationship graph of users. 
The choice of an initial set of users is critical here because these
initially selected users will decide the expected number of users that 
would get influenced ultimately.
An important problem in this context is the {\em influence maximization problem\/} - 
given a social network graph and influence probabilities
on each edge, how do we select a small set of initial users so as
to maximize the number of users who get influenced. This problem has been
solved in the literature as an optimization problem by assuming a suitable model
of information diffusion.
Various algorithms have been proposed to find
the set of influential nodes in a social network efficiently
\cite{KKT,Chen}. 

The current models for information diffusion process assume that
the social planner knows the influence probabilities accurately.
There are many reasons why this assumption is not true in practice.
Consider an example, say a seller wants to sell a newly released book
by some author and another seller wants to sell a tennis racket. Both
the sellers decide to use viral marketing on a popular online social
network to market these products. To do this the seller will have
to find the values of influence probabilities on each edge and find 
the influential set of nodes to initiate the cascade process. 
Consider a typical user $u$ and her set of neighbors on the social network.
See Figure 1.
The influence probability of $u$ on each neighbor for a product
say book need not be the same for another product say tennis racket.
This is perhaps because $u$ is a novice tennis player but is good at
English literature and her friends know this, however, this particular
information is not available on this social networking site. Hence
her friends will be influenced more by her recommendation of a novel
to them than a tennis racket. 

\begin{figure}
\includegraphics[scale=0.35]{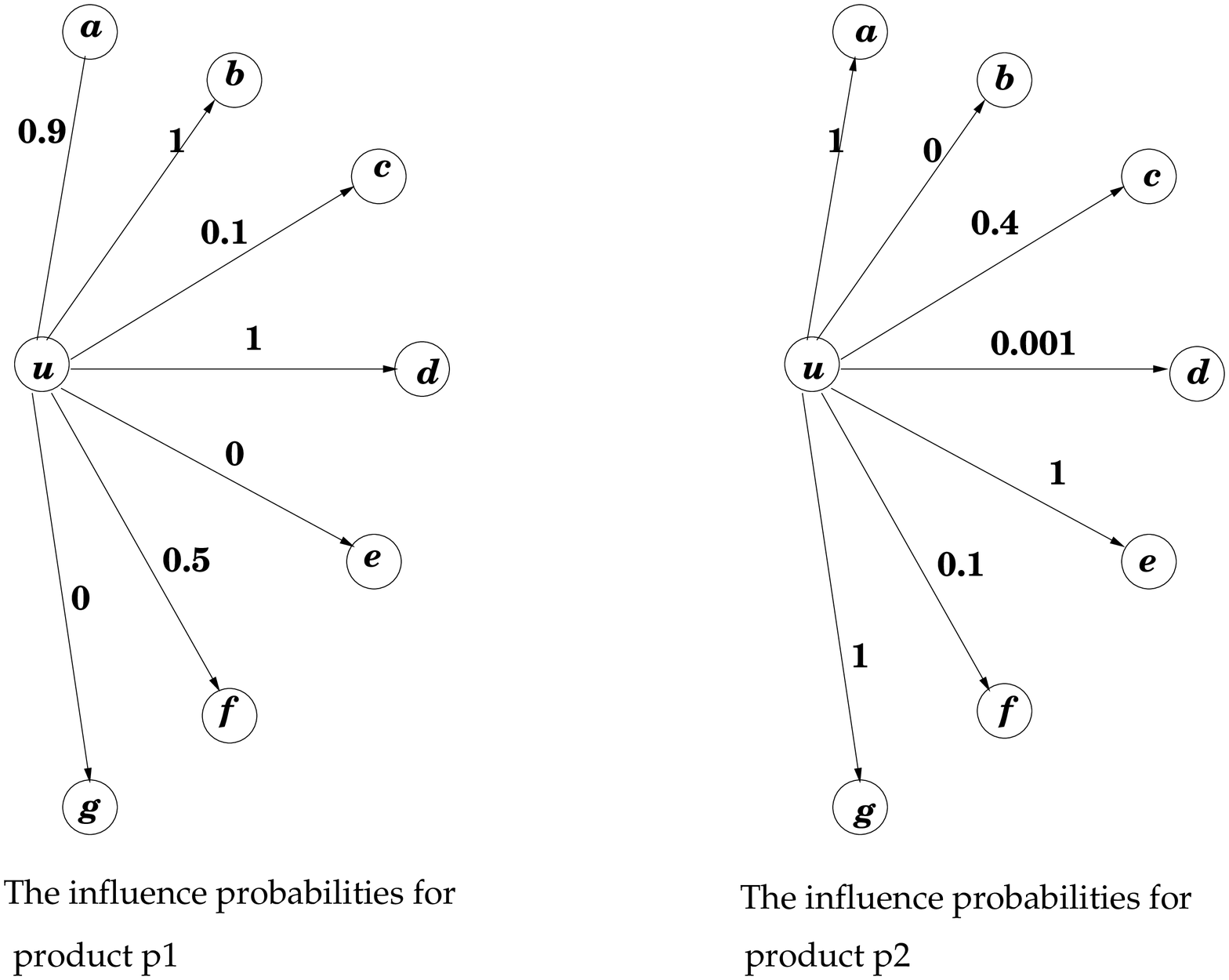}

\caption{Variation of influence probabilities with product}

\end{figure}

Thus, the influence probabilities can vary drastically with the product
being marketed, the time at which the recommendation is made, and
perhaps many other factors. Building robust and accurate models for
estimating the influence probabilities from the currently available
data in  online social networks is a difficult task. The best way to
know the influence probabilities accurately is to elicit them
truthfully from the users themselves.
However, the users need not reveal the influence probabilities 
truthfully due to strategic reasons.  For example, a user might
misreport by projecting a higher level of influence on her friends 
in order to become a part of the initially active set.

The success of a viral marketing strategy or in general
influence maximization is critically dependent on the initially 
chosen target set.  This in turn depends on the accuracy
of the influence probabilities. More generally, predicting
information diffusion in a social network critically depends
on knowing the influence probabilities. 

In this paper we address the influence maximization problem in an incomplete
information setting in which, influence probabilities are the private
information.  Our objective in this paper is to extract influence probabilities 
accurately using a mechanism design approach.

\subsection{Relevant Work}
Kempe, Kleinberg, and Tardos in \cite{KKT} considered the algorithmic
problem of influence maximization proposed by Domingos and Richardson
in \cite{Domingos}. In this paper they proved that this problem is NP-hard
even for simple models of information diffusion and for some of the
more complex models, it is not even constant factor approximable.
They gave a constant factor approximation algorithm for the 
independent cascade model by proving the sub-modularity of the influence
function. The greedy algorithm they propose assumes that the
influence probabilities are available to the algorithm.
There are a number of algorithms proposed in the context of influence
maximization in the recent years; we mention only two papers here:
(a) Leskovec, Krause, and Guestrin \cite{LESKOVEC}
(b) Chen, Wang, and Yang \cite{Chen}. These two papers also contain
a rich set of related references.

Alon, Fischer, Procaccia, and Tennenholtz in \cite{Procaccia} proposed
a game theoretic model for truthfully choosing the agents that maximize
the sum of indegrees in a directed graph. In the model they propose,
they consider the outdegree of an agent as private information. The
objective of each agent is to be among the set of nodes chosen by
the algorithm. 
They propose several deterministic
and randomized strategy-proof algorithms to achieve the objective
of maximizing the sum of indegrees. The problem that we address in
the present paper is different and more general.
The objective of each agent in \cite{Procaccia} is to maximize the 
number of neighbors that she is able to activate in the 0-1 cascade process. 
In our problem, on the other hand,
the objective is to choose a set of agents that have maximum reachability.

A mechanism design based framework to extract the information from
the agents has been proposed for ranking systems \cite{Altman}.
The authors study incentives in ranking systems, where agents try to
maximize their position in the ranking, rather than to obtain a correct outcome. 
They consider several basic properties of ranking systems and  characterize 
the conditions under which incentive compatible ranking systems exist. 
They show  that in general no such system satisfying all the properties exists.

Dixit and Narahari in \cite{Dixit} proposed query
incentive networks in which the nodes, along with the answer are aware
of the quality of the answer. They proposed a game theoretic model
of query incentive networks in which the quality of answer is the
private information. They designed a scoring rule based mechanism
to truthfully extract the quality of the answers from the agents along
with the actual answer. In our work, we design the influencer-influencee
model which is similar to the game theoretic model presented in \cite{Dixit}
for query incentive networks. In \cite{Dixit}
the rewards to the agents depend on the truthfulness of the quality
of answers they report. In our problem, the payments to the agents
depend on the truthfulness of the influence probabilities they report.

In the work by Goyal, Bonchi, and
Lakshmanan in \cite{Goyal}, the approach is to use a machine
learning based approach for predicting the influence
probabilities in social networks. Intuitively, the approach they consider
is that if an individual $x$ takes a total of $a$ actions out of
$b\leq a$ actions were performed by its neighbor $y$
before $x$, then there is a probability of $\frac{b}{a}$ that person
$x$ will be influenced by $y$ in future. Here the {}``action''
is the act of joining a community or group in a social network which
does not involve any effort or monetary transfer. 
They validate the models they build on a real world data set.

To the best our knowledge, the model presented in this paper is the
first one that captures strategic behavior of agents in the information
diffusion process. Using this model, our aim is to elicit the true
influence probabilities from the agents in order to accurately compute
the set of highly influential nodes or predict the progress of
information cascades.

\subsection{Contributions and Outline}
In this paper, we design mechanisms to extract influence probabilities
truthfully from the users of a social network. 
We first work with a simple model, the {\em influencer model\/}, where we
assume that each user knows the level of influence she has on her
neighbors but this is private information. In the second model,
the {\em influencer-influencee model\/}, which is more realistic, we
determine influence probabilities by combining the probability values
reported by the influencers and influencees.

\subsubsection*{The Influencer Model}
First, we develop a game theoretic model for information diffusion
process in which we ask only the influencer to reveal the influence
probabilities on all outgoing edges. 
We propose a VCG  (Vickrey-Clarke-Groves) mechanism \cite{MASCOLELL} based approach.
We show that, without using money or any payment scheme the ideal influence maximizing algorithm
may not be incentive compatible.  Then we show how a VCG mechanism based approach could be used.

\subsubsection*{The Influencer - Influencee Model}
In this more general model, given
an edge in the social network, we ask the influencer as well as the
influencee to reveal the influence probability on the edge. This model
is more realistic since, both the persons involved in a connection,
will have information about the influence probability. 
We design a payment scheme
in which we use scoring rules \cite{Selten}
to design the payment scheme. We show
that it is a Nash equilibrium to report true influence probability
in this mechanism. We also design the reverse weighted scoring rule
(derived from the weighted scoring rule)  which has several
desirable properties which, standard scoring rules like the quadratic
and spherical scoring rules do not possess. 

\subsubsection*{Outline of the Paper}
The rest of the paper is organized as follows.
In the next section (Section 2), we provide essential preliminaries. 
We discuss the {\em influencer model\/} in Section 3 and the 
{\em influencer-influencee model\/}
in Section 4. In Section 5, we briefly discuss how the influencer-influencee
model could be implemented in a typical viral marketing scenario.

\section{Preliminaries}
Several probabilistic models have been proposed to model the spread
of information in social networks. We give here a
brief overview of the \emph{Independent Cascade Model}. We will represent
any social network by directed graph $G$. We will say that an individual
(a node in the graph) is active if she is the adopter of the innovation or the behavior
and inactive otherwise. We will assume that once an individual becomes
active, she cannot switch back to being inactive.

\subsection{Independent Cascade Model}

Let $V = \{1,2, \ldots, n\}$ be the set of nodes in the social
network.
In this model we activate some subset of nodes $A$ initially. The
information diffusion process unfolds in discrete time steps 
($x = 1, 2, \ldots, $) as follows.
Each node that becomes active at time step $x$ will try to activate
each of its neighbors. A node will get the chance to activate its
neighbors only once, that is at the time instant in which it becomes
active. A node $i$ will successfully activate its neighbor $j$ with
probability $p_{ij}$. Thus $p_{ij}$ is the probability that node
$i$ will activate node $j$ conditioned on the event that node $j$
is inactive when node $i$ got activated. The successfully activated
nodes at time instant $x$ can now activate their inactive neighbors at time
instant $x+1$. This process ends when no more nodes get activated.
Clearly this process will end in at most $n-1$ time steps, where $n$ is
the number of nodes.

\subsection{Influence Maximization Problem}

First we define the notion of an influence function.
Given an initially active set $A$ that is a subset of $V$ and influence probabilities, the
influence function denoted by $\sigma(A)$ is the expected number
of active nodes at the end of the diffusion process. 
The influence maximization problem is, given a parameter $k$, a social
network graph $G$, and a model of information diffusion, to find a set
of $k$ nodes in $G$ to be activated initially (also known as target
set) such that, it will maximize the influence function $\sigma(A)$.

\subsection{Scoring Rules}
%
A scoring rule \cite{Selten}  is a sequence of scoring functions,
$ S_1, S_2, \ldots , S_t$, such that $S_i$ assigns a score $S_i(z)$
to every $z \in \Delta(T)$ where $\Delta(T)$ denotes the set of probability
distributions on the set $T = \{1,2, \ldots, t\}$. 
Note that $z = (z_1, z_2, \ldots, z_t)$ with $z_i \geq 0$ for $i=1,2, \ldots, t$ and
$\sum_{i=1}^{t} z_i = 1$.
We will consider only real valued scoring rules. Scoring rules are primarily
used for comparing the predicted distribution with the true observed one.
Suppose $w \in \Delta(T)$ is the true observed distribution.
Then the \textit{expected score} of any general distribution
$z \in \Delta(T)$ against $w$ is defined as 
\[V(z|w) = \sum_{i=1}^t w_iS_i(z)\].
The \textit{expected score loss} is defined as 
\[L(z|w) = V(w|w) - V(z|w)\]
A scoring rule $S_1, S_2, \ldots, S_t$ is called \textit{proper} or \textit{incentive compatible}
if $\forall z,w \in \Delta(T)$ with $z \neq w$, $L(z|w) > 0$.
If the scoring rule is proper, then
it is a best response for each agent to report its true probability
distribution. 
The following are  popular proper scoring rules discussed in the literature \cite{Selten}:
\begin{itemize}
\item Quadratic scoring rule: \[S_i(z_1, \ldots, z_i, \ldots, z_t) = 2z_i - \sum_{j=1}^t z_j^2\]
\item Logarithmic scoring rule: \[S_i(z_1, \ldots, z_i, \ldots, z_t) = ln ~~ z_i \]
\item Spherical scoring rule: \[S_i(z_1, \ldots, z_i, \ldots, z_t) = \frac{z_i}{\sqrt{\sum_{j=1}^t z_j^2}}\]
\item Weighted scoring rule
\[S_{i}(z)=\frac{2i\cdot z_{i}-\sum_{j=1}^{t}z_{j}^{2}\cdot j}{t}\]
\end{itemize}

\section{The Influencer Model}
In this model we assume that only the influencer knows the
influence probabilities and only the influencer is asked to report
the probability values. The model is as follows. 

\begin{itemize}
\item The social planner has the entire graph structure $(V,E)$
of the social network.  The players in the game $V=\{1, \ldots, n\}$ are the users of the social
network.
\item 
Here each agent $i$ has the influence probability vector $\theta_{i}\in[0,1]^{\left|E\right|}$as
her private information. The $jth$ component of $\theta_{i}$ will
give the influence probability of $i$ on node $j$. 
Let $out(i)=\{j\in V|(i,j)\in E\}$ denote the set of successors of node $i$.
The social planner does not
know anything about the influence probability of node $i$ on its
successors.  For all other nodes,
(non-neighbors of node $i$), the influence probability will be
zero and this is known to the social planner as the social planner
has the structure of the graph. 
Prior to starting the information diffusion
process, the social
planner asks each agent to report her influence probabilities. 
The reports from the agents may or may not be truthful.
\item Given the reported influence probabilities, the social planner now
computes the target set using an influence maximization algorithm.
Let this target set be $A$. 
\item Let $\theta\in[0,1]^{\left|E\right|}$ be the true influence probability
vector, representing the influence probability on each edge of the
graph. Then the utility of a player when the social planner chooses
a target set $A$ is  the expected number of neighbors
activated by that player. Also, without involving payments, the valuation
function for each agent is equal to its utility that is, \[
u_{i}(\theta,A)=v_{i}(\theta,A)\]
Thus utility is proportional to the expected number of neighbors an agent is able to activate, 
given the target set. 
\end{itemize}
Here the valuation function represents the preferences of the agents
over the target set chosen by the social planner.

\subsection{A VCG Mechanism Based Approach}
Consider the exact influence maximization problem which optimizes
the influence function. It has been proven in \cite{KKT} that this problem
is NP-hard. Even if we are given that the algorithm to choose
the target set finds the optimal solution,  we
can come up with an example in which the agents would
prefer to lie about their preferences and still be better off. One such example is given
in the next subsection. We can make this algorithm incentive compatible
by introducing appropriate incentives to the agents. 
An immediate and natural approach is to use the
VCG (Vickrey-Clarke-Groves) mechanisms \cite{MASCOLELL}. We shall, in particular,
explore the use of the Clarke payment rule \cite{CLARKE}
to make this mechanism incentive compatible. 
The utility of the agents with Clarke payments  will take the form:\[
u_{i}(\theta,A)=v_{i}(A)+p_{i}(v_{1},v_{2},...,v_{n})\]
where $p_{i}$ is the discount offered to agent $i$ by the social
planner. 

In order to design a VCG mechanism, we first need to prove that the social choice function 
is allocatively efficient.  Note that, in our framework, the social choice function 
is precisely the algorithm being used to choose the target set.
We will first prove the following useful lemma which will immediately
imply that the exact influence maximization algorithm is allocatively
efficient. To prove the lemma, we will first describe an equivalent
view of the independent cascade model given by Kempe, Kleinberg, and
Tardos in \cite{KKT}. 

The influence probability of node $i$ on $j$ denoted by $\theta_{ij}$
gives the probability that node $i$ will activate $j$ given that
node $j$ will be inactive at the instant when node $i$ becomes active.
This event can be viewed as the flip of a biased coin. In the process,
let us say we flip the coins on all the edges before the start of the cascade
process and check the result of the coin flip only when a node becomes
active and its neighboring node is inactive. This change will not
affect the final result and it is equivalent to the original cascade
process. We call the edges on which the coin flip resulted in heads
as live edges and the remaining edges as blocked. Given this equivalent
view, if we fix the outcomes of all coin flips and initially active
set of nodes $A$, then we will get a graph in which some edges are
live and the rest are blocked, depending on the outcome. Clearly in this
graph, if we run the cascade process, then the number of nodes that
are active will be the number of nodes that are reachable from set
$A$ on a path that consists of only live edges. 

Thus we will consider a sample space $S$ in which each sample point
corresponds to one possible outcome of all the coin flips. If $X$
denotes one such fixed outcome of coin flips, we define $\sigma_{X}(A)$ to be the number
of active nodes at the end of the cascade process for the fixed outcome
$X$ and target set $A$. Then $\sigma(A)$, the expected number of active nodes at the end of the
cascade process, is given by:\[
\sigma(A)=\sum_{X\in S}P[X]\sigma_{X}(A)\]
Given this formula for $\sigma(A)$, we can now prove the following
lemma.
\newtheorem{lemma}{lemma}
\begin{lemma}
Given a target set $A$, then\[
\sigma(A)=\sum_{i=1}^{n}v_{i}(\theta,A)+\left|A\right|\]
where $v_{i}$ is the valuation function of agent $i$ which is equal
to the expected number of neighbors activated by that agent.\end{lemma}

\begin{proof}
Fix a sample point $X$ from the sample space $S$ of all possible
coin flips on the edges. Consider an arbitrary node $i$ in the
graph. Let $v_{i_{X}}(A)$ be the number of neighbors activated by
node $i$ for outcome $X$. More concretely, let us define $d(A,v)$ where
$A\subseteq V$ and $v\in V$ as the shortest path distance between
$A$ and a node $v$. Thus, $d(A,v)=0$ if $v\in A$. Let, $R_{Av}=\{u\in V|(u,v)\in 
E_{X} \wedge d(A,u)+1=d(A,v)\}$
where, $E_{X}$ is the edge set that is active for the outcome $X$.
$R_{Av}$ is the set of nodes that are lying on the shortest path
from set $A$ to node $v$. Thus, we define \[
v_{i_{X}}(A)=\left|\{v\in N(i)|i\in R_{Av}\wedge i\succeq j\,\forall j\in R_{Av}\}\right|\]
where, the ordering $\succeq$ is lexicographic (we  are breaking
the ties in favor of the node with the highest lexicographic order).
Also lexicographic ordering ensures that a node is activated deterministically by exactly
one node for a fixed outcome $X$.  Then we have, 
\[ \sigma_{X}(A)=\sum_{i=1}^{n}v_{i_{X}}(A)+\left|A\right|\]
Also we can see that 
\[ \sum_{i=1}^{n}v_{i}(\theta,A)+\left|A\right|=\sum_{X\in S}P[X]\{\sum_{i=1}^{n}v_{i_{X}}(A)+\left|A\right|\}\]
Since $\sum_{X\in S}P[X]=1$, we have
\[
\sum_{i=1}^{n}v_{i}(\theta,A)+\left|A\right|=\sum_{X\in S}P[X]\sigma_{X}(A)\]
This implies
\[
\sum_{i=1}^{n}v_{i}(\theta,A)+\left|A\right|=\sigma(A)\]
\end{proof} 
Thus, by using Lemma 1, we can say that an influence maximization
algorithm that finds the exact optimal solution is allocatively efficient and hence VCG payments will give
us strategy proof mechanism for this algorithm \cite{MASCOLELL}. Thus with VCG payments,
the utility for each agent will be :\[
u_{i}(\theta,A)=v_{i}(A)+\sum_{j\neq i}v_{j}(A)-h(v_{1},...v_{i-1},v_{i+1},...v_{n})\]
Here $h$ is some function independent of $v_{i}$. Thus we have the
following result:
\newtheorem{theorem}{Theorem}
\begin{theorem}
The influence maximization algorithm that finds an exact optimal
solution is allocatively efficient
and hence is dominant strategy incentive compatible under VCG payments.
\end{theorem}

\begin{figure}
\caption{A stylized social network}
~~~

~~~~~~~~~~~~~~~\includegraphics{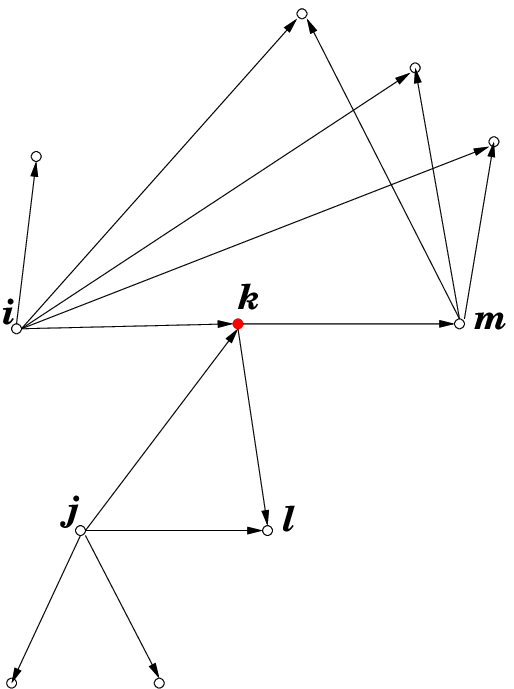}
\end{figure}

\subsection{An Illustrative Example}
We now provide a simple example to illustrate how the proposed model
functions. 
Consider a simple social network graph as shown in Figure
2.
Assume that the true influence probabilities are all $1$ in the graph
and the algorithm for target set selection is an exact influence
maximizing algorithm. Note that all the agents in the network have
complete information about the network. If we want to choose only one node as the target
set, then clearly node $j$ will be chosen, because $\sigma(\{j\})=8$,
which represents the maximum influence among all nodes present, as can be
seen from the figure. Now assume that, all nodes except node $k$
report their true influence probabilities. Consider node $k$, if this node reports
its true influence, then node $j$ will be chosen as target set and
its utility will be $1$, because she will only be able to influence
node $m$. This is because when the cascade reaches node $k$ at time
$x=1$, then till that time its neighbor node $l$ would have already been
influenced by node $j$ at time $t=0$. Now if node $k$ lies about
its influence on node $m$ as $0$, then node $i$ will be chosen
as the influence maximizing target set. For this target set, the utility
of agent $k$ will be $2$ but now the influence function
will be $\sigma(\{i\})=7$.
Thus agent $k$ is better off by lying rather than reporting the truth.
Hence the exact influence maximization algorithm is not incentive compatible
without payments. 

Now assume that we include the Clarke payment scheme
in this scenario and $h(v_{1}, \ldots, v_{i-1},v_{i+1}, \ldots, v_{n})=0$. If
node $j$ is chosen as the target set then agent $k$ will be able
to influence agent $m$ and she will get a monetary payment of $7$
units, thus total utility will be $8$ and if any other node is chosen
as the target set, then node $k's$ payoff will be lower. 

\noindent
{\em Note}: It may be noted that the greedy algorithm proposed by Kempe, Kleinberg,
and Tardos in \cite{KKT} is the same as the exact influence maximization
algorithm in the special case $|A| = 1$. Thus the above example also shows that the
greedy algorithm
is not incentive compatible. We can construct examples for other heuristic
based algorithms like the high degree heuristics \cite{KKT}, 
degree discount heuristics
\cite{Chen} etc., to illustrate that they are
not incentive compatible. All these algorithms use the information reported
by the agents directly to select the target set. Thus, none of these
algorithms are strategy-proof. However the randomized algorithm in
which we select the nodes in the target set uniformly randomly is
strategy-proof. But the expected influence of this algorithm is low
as shown by the experiments done in \cite{KKT}.


\section{Influencer - Influencee Model}
The obvious limitation of the {\em influencer model\/} discussed
above is the restricted assumption that the influencer alone
decides the influence probabilities.
In a real world social network, given a social connection between
two individuals, both the individuals will have information about different
aspects and properties of the connection. The influencer-influencee model
tries to leverage this fact in designing an incentive compatible mechanism
for eliciting influence probabilities.
An advantage of the mechanisms designed with this model is that 
the agents need not know any information beyond its neighborhood.
We now describe this model and propose a mechanism based on scoring rules
for truthfully eliciting influence probabilities.

\subsection{The Model}
\begin{itemize}
\item Given a directed edge $(i,j)$ in the social network, the social planner
will ask: 
\begin{itemize}
\item agent $i$ (the influencer) to report her influence probability
$\theta_{ij}$ on $j$ and  
\item  agent $j$ (the  influencee)
to report agent $i's$ influence on her. 
\end{itemize}
Thus the social
planner will ask each agent to reveal
the probability distribution over each edge which is incident on it
and which is emanating from it.
We can consider the activation probability on each edge as a probability
distribution over the set $\{active,inactive\}$.
\item Given these influence probabilities, the social planner will compute:
\begin{itemize}
\item the influence maximizing target set using an influence maximization
algorithm and 
\item the amount of discount
to be given to the agents based on their reported probability distribution
on edges using a scoring rule based approach that will be
described soon.
\end{itemize}
\item Consider an agent $i$. Let $out(i)=\{j|(i,j)\in E\}$ and
$in(i)=\{j|(j,i)\in E\}$. Thus agent $i$ acts as influencer
to nodes in the set $out(i)$ and acts as the influencee
for the nodes in set \\ $in(i)$. In this model an assumption
is that agent $i$ knows  the influence probabilities on the
edges that are incident on $i$ and that are emanating from $i$.
Thus agent only knows about the influence probabilities in its neighborhood
and nothing beyond that.
\item Also no agent knows what influence probability is reported by
the agents in its neighborhood. The only way an agent can predict
the reported probability by its neighbor is by her own assessment
of it. Thus we assume that for any given pair of nodes $i$ and $j$
having edge $(i,j)$ between them,  the conditional probability
distribution function $P(\theta_{ij}^{j}|\theta_{ij}^{i})$ which
has all the probability mass concentrated at $\theta_{ij}^{j}=\theta_{ij}^{i}$.
\item Here we discretize the continuous interval {[}0,1{]} into $\frac{1}{1+\epsilon}$
equally spaced numbers and agents will have to report the influence
probability by quoting one of the $\frac{1}{1+\epsilon}$ numbers. More concretely,
given set $T=\{1,2, \ldots, t\}$ we define $z\in\{0,\epsilon,2\epsilon,...,1\}^{t}$
such that $\sum_{i=1}^{t}z_{i}=1$. For the case of our problem, $T=\{active,inactive\}$,
thus agents will only have to report one number $\theta_{ij}\in\{0,\epsilon,2\epsilon,...,1\}$.
\end{itemize}
Based on this model we will now design a scoring rule based payment
schemes. 

\subsection{A Scoring Rule Based Mechanism}
We now consider a scoring rule based payment scheme in which we incentivize
agents for reporting the true probability distribution on each edge.
In this mechanism, the payment to an agent $i$ depends on the truthfulness
of the distribution she reveals on edges incident on $i$ as well
as on the edges emanating from $i$. 

First we state a lemma without proof.
The lemma quantifies the amount of loss that an agent suffers
by not reporting its true type.

\begin{lemma}
If $w,z\in\{0,\epsilon,2\epsilon,...,1\}^{t}$, $0<\epsilon\leq1$
such that $\sum_{i=1}^{t}w_{i}=1$ and $\sum_{i=1}^{t}z_{i}=1$ and
$z_{i}=w_{i}\pm\epsilon$ for at least one integer $1\leq i\leq t$, then
\begin{itemize}
\item For quadratic scoring rule \[
V(z|w) \leq V(w|w)-2\epsilon{}^{2}\]

\item For the spherical scoring rule\[
V(z|w) \leq V(w|w)-1.5\epsilon{}^{2}\]

\item For weighted scoring rule\[
V(z|w) \leq V(w|w)-\epsilon{}^{2}\]
\end{itemize}
\end{lemma}

We develop the mechanism assuming the quadratic scoring
rule. A similar development will follow for other
proper scoring rules. In the proposed mechanism, the payment
received by an agent $i$ is given by 
\[ \left(v_{i}(A,\theta)+\frac{d_{i}^{2}}{2\epsilon^{2}}\right)\Biggl(\sum_{j\in in(i)}V_{ji}^{i}(\hat{\theta_{ji}^{j}}|\hat{\theta_{ji}^{i}})+
\sum_{j\in out(i)}V_{ij}^{i}(\hat{\theta_{ij}^{j}}|\hat{\theta_{ij}^{i}})\Biggr)\]
where $d_{i}$ is the degree of agent $i$, $V_{ij}^{i}()$ is the
expected score that agent $i$ gets for reporting the distribution
$\hat{\theta_{ij}^{i}}$ on the edge $(i,j)$. We are now in a position
to state and prove the main result of this paper. The theorem
specifically mentions quadratic scoring rule for the sake of convenience
but will hold for any proper scoring rule.
\begin{theorem}
\emph{Given the influencer-influencee model, reporting true probability distributions is a Nash equilibrium
in a scoring rule based mechanism with quadratic scoring rule.}\end{theorem}
\begin{proof}
We will first consider the strategic behavior of any arbitrary agent
$i$ considering only the agents in the set $S=in(i)\cup out(i)$.
In the payment scheme, the agents belonging to the set $V\setminus S$
can only affect the valuation $v_{i}(A,\theta)$. The expected payoff
an agent $i$ gets is given by\\
\[
\sum_{j\in in(i)}\sum_{\theta_{ji}=0}^{1}P(\theta_{ji}^{j}|\theta_{ji}^{i})v_{i}(A,\theta)V_{ji}^{i}(\hat{\theta_{ji}^{j}}|\hat{\theta_{ji}^{i}})+\]
\[
\sum_{j\in out(i)}\sum_{\theta_{ij}=0}^{1}P(\theta_{ij}^{j}|\theta_{ij}^{i})v_{i}(A,\theta)V_{ij}^{i}(\hat{\theta_{ij}^{j}}|\hat{\theta_{ij}^{i}})+\]
\[
\frac{d_{i}^{2}}{2\epsilon^{2}}\Biggl(\sum_{j\in in(i)}\sum_{\theta_{ji}=0}^{1}P(\theta_{ji}^{j}|\theta_{ji}^{i})V_{ji}^{i}(\hat{\theta_{ji}^{j}}|\hat{\theta_{ji}^{i}})+\]
\[
\sum_{j\in out(i)}\sum_{\theta_{ij}=0}^{1}P(\theta_{ij}^{j}|\theta_{ij}^{i})V_{ij}^{i}(\hat{\theta_{ij}^{j}}|\hat{\theta_{ij}^{i}})\Biggr)\]
Note that the valuation $v_{i}(A,\theta)$ is dependent on the assessment
of the influence probabilities by agent $i$ in its neighborhood. 
Every agent will now try to maximize the expected payoff by considering
the strategies of agents in its neighborhood. If all the agents in the
neighborhood are truthful then we have\[
\sum_{j\in in(i)}\sum_{\theta_{ji}=0}^{1}P(\theta_{ji}|\theta_{ji}^{i})v_{i}(A,\theta)V_{ji}^{i}(\theta_{ji}|\hat{\theta_{ji}^{i}})+\]
\[
\sum_{j\in out(i)}\sum_{\theta_{ij}=0}^{1}P(\theta_{ij}|\theta_{ij}^{i})v_{i}(A,\theta)V_{ij}^{i}(\theta_{ij}|\hat{\theta_{ij}^{i}})+\]
 \[
\frac{d_{i}^{2}}{2\epsilon^{2}}\Biggl(\sum_{j\in in(i)}\sum_{\theta_{ji}=0}^{1}P(\theta_{ji}|\theta_{ji}^{i})V_{ji}^{i}(\theta_{ji}|\hat{\theta_{ji}^{i}})+\]
\[
\sum_{j\in out(i)}\sum_{\theta_{ij}=0}^{1}P(\theta_{ij}|\theta_{ij}^{i})V_{ij}^{i}(\theta_{ij}|\hat{\theta_{ij}^{i}})\Biggr)\]
Now consider the expression
\[
\left(v_{i}(A,\theta)+\frac{d_{i}^{2}}{2\epsilon^{2}}\right)\Biggl(\sum_{j\in in(i)}V_{ji}^{i}(\theta_{ji}|\hat{\theta_{ji}^{i}})+ \sum_{j\in out(i)}V_{ij}^{i}(\theta_{ij}|\hat{\theta_{ij}^{i}})\Biggr)\]
Now Let \[ \beta=\sum_{j\in in(i)}V_{ji}^{i}(\theta_{ji}|\theta_{ji})+ 
\sum_{j\in out(i)}V_{ij}^{i}(\theta_{ij}|\theta_{ij})\]
Thus $\beta$ is the expected score agent $i$ gets when she reports
the true distribution over all the edges. Also, by probability mass
assumption, this is the expected score she will receive when she reports
truthfully. 

Thus, if an agent is truthful then, she will receive a payoff given
by\[
\left(v_{i}(A,\theta)+\frac{d_{i}^{2}}{2\epsilon^{2}}\right)\beta\]
Let $v_{i}(A,\theta)$ be the true valuation of an agent and $v'_{i}(A,\theta)$
the valuation when agent lies. That is when agent reports some
$\hat{\theta_{ij}^{i}}\neq\theta_{ij}$. Consider the utility of an
agent when she lies on only one of the edges\[
\left(v'_{i}(A,\theta)+\frac{d_{i}^{2}}{2\epsilon^{2}}\right)\left(\beta-2\epsilon^{2}\right)\]
When an agent reports a probability value that is $\pm\epsilon$ away
from the true probability value, the quadratic scoring rule (Lemma 2) ensures that
the agent gets payoff lower by $2\epsilon^{2}$ for that edge. Now, we
assume the worst case scenario in which agent gains maximum by reporting
the false probability value on only a single edge and that too minimum
possible deviation from the true value, Since we divide the {[}0,1{]}
probability interval into $1/\epsilon$ numbers, agent $i$  has to report
the probability value that is at least $\epsilon$ away from the true
probability value. That is, agent $i$ will have to report some probability
value $\hat{\theta_{ij}^{i}}=\theta_{ij}\pm\epsilon$. This gives
us $\left(\beta-2\epsilon^{2}\right)$. 

An agent cannot get a higher score for lying, as the quadratic scoring rule is incentive
compatible. Thus the agent can gain only in the valuation part. Now by
reporting a false probability distribution, the agent can get a valuation greater than true
valuation by say $\delta$ that is,\[
v'_{i}(A,\theta)=v_{i}(A,\theta)+\delta\]
where $\delta\in[0,\sum_{j\in out(i)}\theta_{ij}-v_{i}(A,\theta)]$. 
Thus by reporting a false probability value, an agent gets a payoff of\[
\left(v_{i}(A,\theta)+\delta+\frac{d_{i}^{2}}{2\epsilon^{2}}\right)\left(\beta-2\epsilon^{2}\right)\]

Thus for the agent to remain truthful we require that\[
\left(v_{i}(A,\theta)+\frac{d_{i}^{2}}{2\epsilon^{2}}\right)\beta\geq\left(v_{i}(A,\theta)+\delta+\frac{d_{i}^{2}}{2\epsilon^{2}}\right)\left(\beta-2\epsilon^{2}\right)\]
\[
\Longrightarrow d_{i}^{2}+\left(v_{i}(A,\theta)+\delta\right)2\epsilon^{2}\geq\delta\beta\]

As the agents are rational $d_{i}\geq\delta\geq0$. Also, for the quadratic scoring rule the maximum possible expected payoff for each edge will be at most 1. Thus, $\beta\leq d_{i}$.
Now, we have $\delta\beta\leq d_{i}^{2}$ and $\left(v_{i}(A,\theta)+\delta\right)2\epsilon^{2}\geq0$.
Thus, it is a best response strategy for an agent to report truthfully
when the agents in its neighborhood are truthful. 

We still have to resolve the case for agents belonging to the
set $V\setminus S$. Note that these agents only affect the valuation of
agent $i$. In the above analysis, the best response strategy for an
agent was derived by assuming that with minimum possible deviation
from reporting true probability values, agent $i$ is able to gain
maximum possible valuation. Thus even if agents in $V\setminus S$
report any arbitrary values, a best response strategy for agent $i$ is to report truthfully.
Thus without knowing the strategy of agents in set $V\setminus S$,
the best response for an agent $i$ is to report truthfully. Thus
reporting true influence probabilities is a Nash equilibrium in this
mechanism. 
\end{proof}
Note that the social planner will have to choose the value of $\epsilon$
which will decide the accuracy of the probability values extracted
from the users. Smaller the value of $\epsilon$, greater the payment
the seller will have to make to the users. The main advantage of this
mechanism is that the seller can use any of the algorithms to select
the target set. The agents will be truthful regardless of which target
set is chosen. 

A similar payment scheme will work with other scoring rules namely
logarithmic, spherical, and weighted scoring rule. We omit this
due to constraints of space.

\subsection{The Reverse Weighted Scoring Rule}
Standard proper scoring rules such as quadratic, logarithmic, spherical,
and weighted scoring rules have a serious limitation in the
current context.  If the influence
probability on an edge is zero, all  these scoring rules will give
an expected score of 1. Thus, if the social network is the
empty graph in which all the edges are inactive, these standard payment
schemes will give maximum possible expected score. We now propose
the following {\em reverse weighted scoring rule\/} to overcome the
above limitation:
\[ S_{i}(z)=2z_{i}(t-i)-\sum_{j=1}^{t}z_{j}^{2}(t-j)\]
We now show that this scoring rule is proper.

\begin{lemma}
The reverse weighted scoring rule is a proper scoring rule.
\end{lemma}

\begin{proof}
Consider the expression for the expected score of reverse weighted
scoring rule, given distributions z and w:\[
V(z|w)=\sum_{i=1}^{t}w_{i}S_{i}(z)\]
\[
=\frac{2\sum_{i=1}^{t}w_{i}z_{i}(t-i)-\sum_{i=1}^{t}z_{i}^{2}(t-i)}{t}\]
\[
V(w|w)=\sum_{i=1}^{t}w_{i}S_{i}(w)\]
\[
=\frac{2\sum_{i=1}^{t}w_{i}^{2}(t-i)-\sum_{i=1}^{t}w_{i}^{2}(t-i)}{t}\]
\[
=\frac{2\sum_{i=1}^{t}w_{i}^{2}(t-i)}{t}\]
\[
L(z|w)=V(w|w)-V(z|w)\]
\[
=\frac{\sum_{i=1}^{t}(w_{i}-z_{i})^{2}(t-i)}{t}\]
since, $i\leq t$ the loss is always positive. Thus, reverse weighted
scoring rule is proper.
\end{proof}

It can also be shown that the the reverse weighted scoring rule 
also satisfies the following desirable properties: 
\begin{enumerate}
\item The expected score is proportional to the influence probability.
\item If $\theta_{ij}=0$ then the expected score for the edge $(i,j)$
to both the agents $u$ and $v$ is zero. That is, $V_{ij}^{i}(\theta_{ij}^{j}|\theta_{ij}^{i})=V_{ij}^{j}(\theta_{ij}^{i}|\theta_{ij}^{j})=0$
if $\theta_{ij}=0$.
\end{enumerate}
Property 1 is desirable because the social planner would want to reward
the agent which revealed the social connection through which the product
can be sold with high probability. Property 2 ensures that an agent
does not get anything for revealing a social connection through which
the product cannot be sold.

\section{Implementation in  Viral Marketing Scenarios}
Since the mechanisms presented under the influencer-influencee
model involve payments, the influence maximization process will
involve monetary transfers.  We now discuss the implementation of the mechanism in the context
of viral marketing in an online social network like facebook, orkut,
etc. 

Consider that a seller wants to market a certain product through a
given social
network. The seller can now ask each user in the social network to
reveal her influence on each of the users in their friends list. Users
have incentive to participate in this mechanism because each user
will get a certain positive payment based on the influence probabilities
they report. The seller can ask each agent to report the influence
probability by developing the application on the social network. There
are a large number of applications on the social networks like orkut,
facebook etc. which users use extensively for playing games and socializing
online. 

In an online social network like facebook, for example, if a user is
interested in the product to be sold, then she can grant the access
to the application. Now, this application will have full access to
the friends list and other profile information of the user that is
public. Thus such an application can be implemented in an online
social network without any privacy issues. The seller will first have
to fix the level of accuracy that she needs before starting the information
extraction process. 

Now, given the influence probabilities, the application will compute
the influence maximizing target set. The application will also compute
the payment to be made to the user. This payment can be made in the
form of a discount (up to 100 percent) on the product or discount coupons
or complimentary gifts or a combination of these.
The marketing and sale of the product will now proceed as in the independent
cascade process and users will get appropriate discounts based on
the reported influence probabilities. Since the payment scheme ensures
incentive compatibility, every user will truthfully reveal the information
to the seller via the application.

\section{Summary and Future Work}
In this paper, we have proposed  mechanisms for
eliciting influence probabilities truthfully in a social network.
Influence maximization in general and viral marketing
in particular are the immediate applications.
The work opens up several interesting questions:
\begin{itemize}
\item In this model we assumed that the influence probability is known exactly to the agents. We can relax this assumption and assume that agents 
know the belief probability rather than exact influence probability.
\item In the influencer model, does there exist an incentive compatible
algorithm having a constant factor approximation ratio with or without
payments?
\item In the influencer model, does there exist a heuristic based incentive
compatible algorithm which may not have theoretical guarantees about
the approximation but in a practical sense it performs well?
\item In the influencer-influencee model, the payments depend on $\epsilon$ which
decides the accuracy of the probability distribution. The higher the
accuracy is required, the higher is the payment to be made to the user.
An interesting direction of future research would be to design 
incentive compatible mechanisms that are independent of this factor. 
\end{itemize}

%
\bibliographystyle{abbrv}
\bibliography{paper}  
\end{document}